\documentclass[conference]{ieeetran}
\IEEEoverridecommandlockouts
\def\BibTeX{{\rm B\kern-.05em{\sc i\kern-.025em b}\kern-.08em
    T\kern-.1667em\lower.7ex\hbox{E}\kern-.125emX}}

\usepackage{graphicx}
\usepackage{amssymb}
\usepackage{amsmath}
\usepackage{amsthm}
\newtheorem{thm}{Theorem}

\newtheorem{lem}{Lemma}
\newtheorem{prob}{Problem}

\newtheorem{assm}{Assumption}
\newtheorem{reform}{Reformulation}

\usepackage[ruled,norelsize]{algorithm2e}

\newcommand{\R}{\mathbb{R}}
\newcommand{\N}{\mathbb{N}}
\renewcommand{\P}{\mathbb{P}}

\newcommand{\pr}[1]{\mathbb{P}\!\left(#1\right)}
\newcommand{\ex}[1]{\mathbb{E}\!\left[#1\right]}

\newcommand{\tr}[1]{\mathrm{tr}\left(#1\right)}
\newcommand{\var}[1]{\mathrm{Var}\!\left(#1\right)}

\newcommand{\std}[1]{\mathrm{Std}\!\left(#1\right)}

\newcommand{\cov}[2]{\mathrm{Cov}\!\left(#1,#2\right)}

\newcommand{\bvec}[1]{\vec{\boldsymbol{#1}}}
\newcommand{\Nt}[2]{\mathbb{N}_{[#1,#2]}}

\newcommand{\vect}[1]{\mathrm{vec}\left( #1\right)}

\makeatletter
\let\NAT@parse\undefined
\makeatother
\usepackage{hyperref}
\hypersetup{
    colorlinks=true,      
    linkcolor=black,
    citecolor=black,
    filecolor=black,
    urlcolor=black     
    }

\title{Chance Constrained Stochastic Optimal Control for Linear Systems with Time Varying Random Plant Parameters}
\author{Shawn Priore, Ali Bidram, and Meeko Oishi
    \thanks{
        This material is based upon work supported by the National Science Foundation under NSF Grant Numbers CMMI-2105631 and OIA-1757207.  Any opinions, findings, and conclusions or recommendations expressed in this material are those of the authors and do not necessarily reflect the views of the National Science Foundation.  
        \newline \indent Shawn Priore, Ali Bidram, and Meeko Oishi are with the Department of Electrical and Computer Engineering, University of New Mexico, Albuquerque, NM; e-mail: \texttt{shawnpriore@unm.edu} (corresponding author)\texttt{, bidram@unm.edu, oishi@unm.edu}.
    }
}

\begin{document}
\maketitle

\begin{abstract}
We propose an open loop control scheme for linear systems with time-varying random elements in the plant's state matrix. This paper focuses on joint chance constraints for potentially time-varying target sets. Under assumption of finite and known expectation and variance, we use the one-sided Vysochanskij–Petunin inequality to reformulate joint chance constraints into a tractable form. We demonstrate our methodology on a two-bus power system with stochastic load and wind power generation. We compare our method with situation approach. We show that the proposed method had superior solve times and favorable optimally considerations. 
\end{abstract}

\section{Introduction} \label{sec:intro}

In much of the linear controls literature, stochasticity is regarded as a factor external to the system modeling process. Additive noise is often a placeholder for systemic uncertainty that is difficult to account for. For example, wind speeds can affect the output of a wind turbine in a local grid, yet state-of-the-art models have considerable difficulty in making accurate predictions of their power output \cite{Chen2010}. New control techniques that can incorporate this stochasticity systemically have the potential to enable more efficient controllers that can be robust to natural phenomena. In this paper, we develop an optimal control derivation scheme for discrete time linear systems with time-varying stochastic elements in the state matrix subject to joint chance constraints.  

Early work in the 1960s and 1970s illuminated the need for incorporating random elements into the plant with applications in industrial manufacturing, communications systems, and econometrics \cite{Aoki1975, Aoki1967, Tou1963}. Several works considered minimization strategies for linear quadratic regulator problems. Without the addition of joint chance constraints, dynamic programming techniques can easily be employed to find optimal controllers \cite{Asher1976, BarShalom1969, Drenick1964}. These works have been extended to account for unknown distributions associated with the random parameters. Sampling techniques and feedback mechanisms have been used to overcome these hurdles \cite{Joseph1961, Oberlin1973}. Unfortunately, these regulation problems are often limited in scope and cannot readily be extended to solve for chance constraints. Random plants with more complex structure have been investigated \cite{Tse1973} but have typically been limited to Gaussian disturbances. Since the late 1970s research in this area has been sparse, appearing only occasionally in econometric literature \cite{BARSHALOM1978, TUCCI2006} where plant uncertainty has been used to model economic trends. 

A similar problem, in which the uncertainty in the plant is modeled either by bounded parameterization or a bounded column space, has been extensively studied in the robust model predictive control community \cite{Wonham1967, nilim2005, Fleming2015, Gravell2021}. By exploiting the bounded parameter and column spaces, estimation  \cite{Li2007, Liu2013}  and stability techniques  \cite{Kouramas2015, Marcos2005} allow for closed loop controller synthesis. While several of these techniques can address uncertainty in the plant, they do not address uncertainty that is random in nature \cite{Daafouz2001, Cuzzola2002}, such as unknown but deterministic parameters. Further, these methods can address uncertainty that result from bounded random variables, such as discrete distributions with finite outcomes, and uniform or beta distributions, but cannot address random variables on semi-infinite or infinite supports. 

We propose to address stochastic optimal control for systems with uncertain state matrices in a manner that is amenable to convex optimization techniques. To achieve this, we use Boole's inequality \cite{casella2002} and the one-sided Vysochanskij–Petunin inequality \cite{Mercadier2021} to transform the chance constraint into a biconvex constraint that can be solved with the alternate convex search method. Our approach offers a closed form reformulation of the chance constraints that is biconvex and can readily be solved. Further, this approach enables optimization under a wide range of distributional assumptions and any solution guarantees chance constraint satisfaction. However, our method also introduces conservatism and relies on open loop controller synthesis. In general, open loop control has known limitations with respect to stability and convergence. As is common in model predictive control literature, this approach could be combined with stabilizing controllers which introduce an extraneous input \cite{Mesbah2016}. The proposed approach accommodates that well established framework which implicitly addresses issues of stabilization and convergence. Hence, many of the known limitations typically associated with open loop control can be accommodated. In addition, there are systems, such as those with limited actuation or sensing, for which feedback is simply not possible \cite{Sudalagunta2018, Koning2019}.   {\em The main contribution of this paper is the construction of a tractable optimization problem that solves for convex joint chance constraints in the presence of random elements in the state matrix.} 

The paper is organized as follows. Section \ref{sec:prelim} provides mathematical preliminaries and formulates the optimization problem. Section \ref{sec:methods} derives the reformulation of the chance constraints with Boole's inequality and the one-sided Vysochanskij–Petunin inequality. Section \ref{sec:results} demonstrates our approach on two problems involving power generation and labor allocation, and Section \ref{sec:conclusion} provides concluding remarks.

\section{Preliminaries and Problem Formulation} \label{sec:prelim}

We denote the interval that enumerates all natural numbers from $a$ to $b$, inclusively, as $\Nt{a}{b}$. Random components will be denoted with bold case, such as $\bvec{x}$ for vectors and $\boldsymbol{A}$ for matrices, regardless of dimension. We use the notation $a_{ij}$ to denote the $(i,j)$\textsuperscript{th} element of the matrix $A$. For a random variable $\boldsymbol{x}$, we denote the expectation as $\ex{\boldsymbol{x}}$, and variance as $\var{\boldsymbol{x}}$, and standard deviation as $\std{\boldsymbol{x}}$. We use ${\coprod_{i=a}^b}$ for when $a>b$ to denote the multiplication of elements over the index $i$ as it decreases from $a$ to $b$ by $-1$.  For a matrix $A$, the operator $\mathrm{vec}(A)$ vertically concatenates the columns of $A$ into a column vector. For two matrices $A$ and $B$, we denote the Kronecker product as $A \otimes B$. For matrix entries $A_1, \ldots, A_m$, we denote a block diagonal matrix constructed with these elements as $\mathrm{blkdiag}(A_1, \ldots, A_m)$. We denote an identity matrix of size $n$ as $I_n$ and the $i$\textsuperscript{th} column of an appropriately sized identity matrix as $\vec{e}_i$.

\subsection{Problem Formulation}

We consider a discrete-time linear system given by
\begin{equation} \label{eq:dynamics}
    \bvec{x}(k+1) = \boldsymbol{A}(k) \bvec{x}(k) + B \vec{u}(k) 
\end{equation}
with state $\boldsymbol{x}(k) \in \mathcal{X} \subseteq \R^n$, input  $\vec{u}(k) \in \mathcal{U} \subseteq \R^m$, and time index $k \in \Nt{0}{N}$. We presume initial conditions, $\vec{x}(0)$, are known, and the set $\mathcal{U}$ is convex. The state matrix $\boldsymbol{A}(k)$ contains real valued random variables, $\boldsymbol{a}_{ij}$, each with probability space $(\Omega, \mathcal{B}(\Omega), \P_{\boldsymbol{a}_{ij}})$ with outcomes $\Omega$, Borel $\sigma$-algebra $\mathcal{B}(\Omega)$, and probability measure $\P_{\boldsymbol{a}_{ij}}$ \cite{casella2002}.

We write the concatenated dynamics as an affine combination of the initial condition and the concatenated control sequence,
\begin{equation} \label{eq:lin_dynamics}
    \bvec{x}(k) = \coprod_{i=k-1}^{0} \boldsymbol{A}(i) \vec{x}(0) +  \boldsymbol{\mathcal{C}}_{k-1} \mathcal{B}\vec{U}
\end{equation}
with 
\begin{subequations}
\begin{alignat}{2}
    \boldsymbol{\mathcal{C}}_k = & \left[
    {\displaystyle \coprod_{i=k}^{1}} \boldsymbol{A}(i) \  \cdots \ \ \boldsymbol{A}(k) \ I_n \ 0_{n \times (N-k-1)n} \right] && \in \R^{n \times Nn}  \\
    \mathcal{B} = &\;  \left( I_N \otimes B \right) && \in \R^{Nn \times Nm} \\
    \vec{U} =& 
    \begin{bmatrix} 
        \vec{u}(0)^\top & 
        \ldots & 
        \vec{u}(N-1)^\top 
    \end{bmatrix}^\top &&\in \mathcal{U}^{N} 
\end{alignat}
\end{subequations} 

\begin{assm}\label{assm:1}
All random components $\boldsymbol{a}_{ij}(k)$ are mutually independent within their matrix. Further, the random matrices $\boldsymbol{A}(k)$ are mutually independent for all time steps.
\end{assm}
\begin{assm}\label{assm:2}
Each random element $\boldsymbol{a}_{ij}(k)$ has a finite expectation and variance. 
\end{assm}

Both assumptions are easily met in most scenarios. We would expect the parameters to be independent in many biological and physical processes, and most distributional assumptions would provide for finite expectation and variance. Of notable exception are certain parameterizations of the $t$, the Pareto, and the inverse-Gamma distributions.  

We presume desired polytopic sets, represented by the linear inequalities $\vec{G}_{ik} \vec{x}(k) \leq h_{ik}$, that the state must stay within at each time step with a desired likelihood
\begin{equation}\label{eq:constraint_t}
        \pr{ \cap_{k=1}^N \cap_{i=1}^{c_k} \vec{G}_{ik} \bvec{x}(k) \leq h_{ik}}  \geq  1-\alpha
\end{equation}
where $c_k$ is the number of linear inequalities. We presume convex, compact, and polytopic sets $ \left\{ \bvec{x}(k) \middle| \cap_{i=1}^{c_k} \vec{G}_{ik} \bvec{x}(k) \leq h_{ik} \right\} \subseteq \mathcal{X}$, and probabilistic violation threshold $\alpha <  1/6$. 

\begin{assm}\label{assm:3}
The distribution describing each probabilistic constraint $\pr{\vec{G}_{ik} \bvec{x}(k) \leq h_{ik}}$ is marginally unimodal. 
\end{assm}

This is likely to be the most restrictive assumption as verifying unimodality can be challenging in cases where the distributional assumptions are not strongly unimodal \cite{Ibragimov1956}. For a thorough review of unimodality in distributions and strong unimodality, we recommend \cite{Bertin1997}. The primary concern for unimodality within this framework is maintaining unimodality through both additive and multiplicative operations. As the terminal time increases the more likely a non-unimodal distribution can arise from the complex and intricate interactions of the random state and the random plant parameters.

We seek to minimize a convex performance objective $J: \mathcal{X}^{N} \times \mathcal{U}^{N} \rightarrow \R$. 
\begin{subequations}\label{prob:big_prob_eq}
    \begin{align}
        \underset{\vec{U}}{\mathrm{minimize}} \quad & J\left(
        \bvec{x}(1), \dots, \bvec{x}(N), \vec{U}\right)  \\
        \mathrm{subject\ to} \quad  & \vec{U} \in \mathcal{U}^N,  \\
        & \text{Dynamics } \eqref{eq:dynamics} \text{ with }
        \vec{x}(0) \label{eq:prob1_dyn}\\
        & \text{Probabilistic constraint \eqref{eq:constraint_t}}  \label{prob:initial_eq_prob_constraints} 
    \end{align}
\end{subequations}

\begin{prob} \label{prob:1}
    Under Assumptions \ref{assm:1}-\ref{assm:3}, solve the stochastic optimization problem \eqref{prob:big_prob_eq} with open loop control $\vec{U}\in  \mathcal{U}^N$, and probabilistic violation threshold $\alpha$.
\end{prob}
The main challenge in solving Problem \ref{prob:1} is assuring \eqref{prob:initial_eq_prob_constraints}. The interaction of multiplying the random state matrices makes enforcing the constraints challenging. Even if closed form expressions exist for a single time step there is no guarantee an expression will exist at the next time step.  

\section{Methods} \label{sec:methods}

Our approach to solve Problem \ref{prob:1} involves reformulating the joint chance constraint \eqref{eq:constraint_t} into a  series of constraints that are affine in the constraint’s expectation and standard deviation, $\ex{\vec{G}_{ik} \bvec{x}(k)}$ and $\std{\vec{G}_{ik} \bvec{x}(k)}$, respectively. This form is amenable to the use of the one-sided Vysochanskij–Petunin inequality which guarantees the synthesized controller satisfies the probabilistic constraint. The reformulation results in an easy to solve biconvex optimization problem. 

\subsection{The Vysochanskij–Petunin Inequality} \label{ssec:vp_i}

The one-sided Vysochanskij–Petunin inequality \cite{Mercadier2021} is the foundational theorem underpinning the approach we take in this work.

\begin{thm}[The one-sided Vysochanskij–Petunin inequality \cite{Mercadier2021}]
Let $\boldsymbol{x}$ be a real valued unimodal random variable with finite expectation $\ex{\boldsymbol{x}}$ and finite, non-zero standard deviation $\std{\boldsymbol{x}}$. Then, for $\lambda > \sqrt{5/3}$, 
\begin{equation} \label{eq:vp}
    \pr{  \boldsymbol{x} - \ex{\boldsymbol{x}}  \geq  \lambda \std{\boldsymbol{x}}} \leq \frac{4}{9(\lambda^2+1)}
\end{equation}
\end{thm}

The one-sided Vysochanskij–Petunin inequality is a refinement of Cantelli's inequality for unimodal distributions. Based on Gauss's inequality, it provides a bound for one-sided tail probabilities of a unimodal random variable being sufficiently far away from the expectation. Specifically, the bound encompasses values at least $\lambda$ standard deviations away from the mean. We first make use of \eqref{eq:vp} to bound the chance constraint probabilities based on an affine summation of the expectation and standard deviation.

\subsection{Constraint Reformulation}

For brevity, we drop the index $i$ by assuming $c_k=1$. We take the complement and employ Boole's inequality \cite{casella2002} to convert the joint chance constraint into a sum of individual chance constraints,
\begin{equation}
    \pr{ \sum_{k=1}^{N} \vec{G}_{k} \bvec{x}(k) \geq h_{k} }
    \leq  \sum_{k=1}^{N} \pr{\vec{G}_{k} \bvec{x}(k) \geq h_{k}}
\end{equation}
Using the approach in \cite{ono2008iterative}, we introduce risk allocation variables $\omega_k$ for each of the individual chance constraints and bound the sum of risk allocation variables,
\begin{subequations}\label{eq:quantile_reform_new_var}
\begin{align}
     \pr{\vec{G}_{k} \bvec{x}(k) \geq h_{k} } &\leq \omega_k  \quad \forall k \in \N_{[1,N]} \label{eq:quantile_reform_new_var_1_1}\\
     \sum_{k=1}^{N} \omega_k &\leq \alpha 
\end{align}
\end{subequations}
where $\omega_k$ is a non-negative real number. 

Here, we need to find an appropriate value for $\omega_k$ such that we can solve this problem. To that end, we add an additional constraint
\begin{equation} \label{eq:add_target}
    \ex{\vec{G}_{k}\bvec{x}(k)} + \lambda_k \std{\vec{G}_{k} \bvec{x}(k)} \leq  h_{k}
\end{equation}
with optimization parameter $\lambda_k > 0$. Here, \eqref{eq:add_target} implies 
\begin{equation}\label{eq:first_bound}
\begin{split}
    &\pr{\vec{G}_{k} \bvec{x}(k) \geq h_{ik} } \\
    & \; \leq \pr{\vec{G}_{k} \bvec{x}(k) \geq \ex{\vec{G}_{k} \bvec{x}(k)} \!+\! \lambda_k \std{\vec{G}_{k} \bvec{x}(k)}} 
\end{split}
\end{equation}
Under Assumption \ref{assm:3}, the one-sided Vysochanskij–Petunin inequality allows us to bound,
\begin{equation}\label{eq:cheby}
    \pr{\vec{G}_{k} \bvec{x}(k) \geq \ex{\vec{G}_{k} \bvec{x}(k)} \!+\! \lambda_k \std{\vec{G}_{k} \bvec{x}(k)} }\leq \frac{4}{9(\lambda_k^2\!+\!1)}
\end{equation}
so long as $\lambda_k \geq \sqrt{5/3}$. Here, $\alpha <1/6$ implies $\lambda_k$ cannot take values smaller than $\sqrt{5/3}$.

By substituting \eqref{eq:first_bound}-\eqref{eq:cheby} into \eqref{eq:quantile_reform_new_var_1_1}, we can establish the relationship between $\lambda_k$ and $\omega_k$ as
$$\omega_k = \frac{4}{9(\lambda_k^2\!+\!1)}$$
Hence, \eqref{eq:quantile_reform_new_var}-\eqref{eq:cheby} simplifies to 
\begin{subequations}\label{eq:quantile_reform_new_var_3}
\begin{align}
     \ex{\vec{G}_{k} \bvec{x}(k)} + \lambda_k \std{\vec{G}_{k} \bvec{x}(k)} & \leq  h_{k} \quad \forall k \in \Nt{1}{N} \label{eq:target_reform}\\
     \sum_{k=1}^{N} \frac{4}{9(\lambda_k^{2}+1)} &\leq \alpha \label{eq:target_lambda}
\end{align}
\end{subequations}
for optimization parameter $\lambda_k > \sqrt{5/3}$.

\begin{lem} \label{lem:1}
For the controller $\vec{U}$, if there exists risk allocation variables $\lambda_{k}$ satisfying \eqref{eq:quantile_reform_new_var_3} for constraints in the form of \eqref{eq:constraint_t}, then $\vec{U}$ satisfies \eqref{prob:initial_eq_prob_constraints}.
\end{lem}

\begin{proof}
Satisfaction of \eqref{eq:target_reform} implies \eqref{eq:first_bound} holds. The one-sided Vysochanskij–Petunin inequality upper bounds \eqref{eq:first_bound} via \eqref{eq:cheby}. Boole's inequality and De Morgan's law \cite{casella2002} guarantee that if \eqref{eq:target_lambda} holds then \eqref{prob:initial_eq_prob_constraints} is satisfied.
\end{proof}

We formally define the reformulated optimization problem.
\begin{subequations}\label{prob:big_prob_eq2}
    \begin{align}
        \underset{\vec{u}, \lambda_1, \dots, \lambda_{N}}{\mathrm{minimize}} \quad & J\left(
        \bvec{x}(1), \dots, \boldsymbol{x}(N), \vec{U} \right)  \\
        \mathrm{subject\ to\ } \quad  & \vec{U} \in \mathcal{U}^N,  \\
        & \text{Expectation and variance derived} \nonumber \\
        & \text{from dynamics } \eqref{eq:dynamics} \text{ with } \vec{x}(0) \label{eq:prob2_dyn}\\
        & \text{Constraint \eqref{eq:quantile_reform_new_var_3}} \label{eq:prob2_constraint} 
    \end{align}
\end{subequations}
\begin{reform} \label{prob:2}
    Under Assumptions \ref{assm:1}-\ref{assm:3}, solve the stochastic optimization problem \eqref{prob:big_prob_eq2} with open loop control $\vec{U} \in  \mathcal{U}^N$, optimization parameters $\lambda_k$, and probabilistic violation threshold $\alpha$.
\end{reform}

\begin{lem} \label{lem:2}
Any solution to Reformulation \ref{prob:2} is a conservative solution to Problem \ref{prob:1}.
\end{lem}
\begin{proof}
By Lemma \ref{lem:1}, \eqref{eq:prob2_dyn}-\eqref{eq:prob2_constraint} satisfy \eqref{prob:initial_eq_prob_constraints}. Here, \eqref{eq:prob2_dyn} replaces \eqref{eq:prob1_dyn} as we only need the expectation and variance derived from the dynamics. All other elements remain unchanged. Conservatism is introduced from Boole's inequality as equality is only achieved when constraints are independent. Similarly, the one-sided Vysochanskij-Petunin inequality only achieves equality only in extremely rare cases. In most scenarios, the use of the Vysochanskij-Petunin inequality will introduce conservatism.
\end{proof}

Here, Reformulation \ref{prob:2} is a conservative but tractable reformulation of Problem \ref{prob:1}. While we cannot guarantee a solution exists to Reformulation \ref{prob:2}, we can guarantee any solution to Reformulation \ref{prob:2} is a solution to Problem \ref{prob:1}, if one exists.

\subsection{Solving Reformulation \ref{prob:2}}

We note that while \eqref{eq:target_reform} will elicit a closed form due to Assumptions \ref{assm:1} and \ref{assm:2}, and linear dynamics, deriving this expression is tedious. This is particularly true for longer time horizons. For the random variable $\vec{G}_{k} \bvec{x}(k)$, the affine form of \eqref{eq:lin_dynamics} allows us to easily compute the expectation via the linearity of the expectation operator, 
\begin{equation} \label{eq:exp_at_k}
    \ex{\vec{G}_{k} \bvec{x}(k)} = \vec{G}_{k} \ex{\coprod_{i=k}^{0} \boldsymbol{A}(i)} \vec{x}(0) +  \ex{\boldsymbol{\mathcal{C}}_k}\mathcal{B}\vec{U}
\end{equation}
Here, $\ex{\coprod_{i=k}^{0} \boldsymbol{A}(i)}$ and $\ex{\boldsymbol{\mathcal{C}}}$ can easily be computed by observing 
\begin{equation}
    \ex{\coprod_{k=a}^b \boldsymbol{A}(k)} = \coprod_{k=a}^b \ex{\boldsymbol{A}(k)}
\end{equation}
for any $a \in \Nt{0}{N-1}$ and $b \in \Nt{0}{N-1}$ by Assumption \ref{assm:1}.

To derive the standard deviation of $\vec{G}_{k} \bvec{x}(k)$, we start by noting five formulas. First, by construction of $\boldsymbol{\mathcal{C}}_k$, there exists some $a \in \Nt{1}{N}$ such that for $j^{\ast} = j - an$ we can write 
\begin{equation} \label{eq:ck_ej}
    \boldsymbol{\mathcal{C}}_k \vec{e}_j =
    \begin{cases}
        \coprod_{i=k}^{a}  \boldsymbol{A}(i) \vec{e}_{j^{\ast}} & \text{if } a \leq k\\
        \vec{e}_{j^{\ast}} & \text{if } a = k+1 \\
        \vec{0} & \text{if } a > k+1
    \end{cases}
\end{equation}
Second, for a random matrix $\boldsymbol{Z}$ following Assumption \ref{assm:1} and a non-random matrix $S$,
\begin{equation} \label{eq:ex_quad}
\begin{split}
     \ex{\boldsymbol{Z}^{\top} \!S \boldsymbol{Z}}
    &  =\ex{\boldsymbol{Z}^{\top}} S \ex{\boldsymbol{Z}} \\
    & \; + \mathrm{diag}\left(\tr{S \var{\boldsymbol{Z} \vec{e}_1}}, \dots, \tr{S \var{\boldsymbol{Z} \vec{e}_n}} \right)    
\end{split}
\end{equation}
The third and fourth formulas are the derived recursive formulas \eqref{eq:var_ay}-\eqref{eq:var_ak_quad}. Here, \eqref{eq:var_ay} is a result of the Law of Total Variance \cite{casella2002}, \eqref{eq:law_exp} is a result of the Law of Total Expectation \cite{casella2002}, \eqref{eq:vec_prop} results from the vectorization function, and \eqref{eq:kron_prop} results from the mixed-product property of the Kronecker product. Finally, for $a < b \in \Nt{0}{k}$

\begin{figure*}
For a known vector $\vec{y}$ and $a \in \Nt{0}{k}$
\begin{subequations} \label{eq:var_ay}
    \begin{align}
    & \var{\coprod_{i=k}^{a} \boldsymbol{A}(i) \vec{y}}\\
    & \; =\ex{\var{\coprod_{i=k}^{a} \boldsymbol{A}(i) \vec{y} \middle| \coprod_{i=k-1}^{a} \boldsymbol{A}(i) \vec{y} }}  + \var{\ex{\coprod_{i=k}^{a} \boldsymbol{A}(i) \vec{y} \middle| \coprod_{i=k-1}^{a} \boldsymbol{A}(i) \vec{y}} } \\
    & \; =\ex{\var{\left( \vec{y}^{\top} \coprod_{i=a}^{k-1} \boldsymbol{A}(i)^{\top} \otimes I_n \right)\vect{\boldsymbol{A}(k)}  \middle| \coprod_{i=k-1}^{a} \boldsymbol{A}(i) \vec{y}}}  + \var{\ex{\boldsymbol{A}(k)}\coprod_{i=k-1}^{a} \boldsymbol{A}(i) \vec{y}} \\
    & \; = \underbrace{\ex{\left( \vec{y}^{\top} \prod_{i=a}^{k-1} \boldsymbol{A}(i)^{\top} \otimes I_n \right) \var{\vect{\boldsymbol{A}(k)} }  \left(  \coprod_{i=k-1}^{a} \boldsymbol{A}(i) \vec{y} \otimes I_n \right) }}_{\text{See }\eqref{eq:var_ak_quad}} + \ex{\boldsymbol{A}(k)}\var{\coprod_{i=k-1}^{a} \boldsymbol{A}(i) \vec{y}} \ex{\boldsymbol{A}(k)^{\top}}
\end{align}
\end{subequations}
\hrulefill
\begin{subequations} \label{eq:var_ak_quad}
\begin{align}
    & \ex{\left( \vec{y}^{\top} \prod_{i=a}^{k-1} \boldsymbol{A}(i)^{\top} \otimes I_n \right) \var{\vect{\boldsymbol{A}(k)} }  \left(  \coprod_{i=k-1}^{a} \boldsymbol{A}(i) \vec{y} \otimes I_n \right) } \\
    & \; = \ex{\ex{\left( \vec{y}^{\top} \prod_{i=a}^{k-1} \boldsymbol{A}(i)^{\top} \otimes I_n \right) \var{\vect{\boldsymbol{A}(k)} }  \left(  \coprod_{i=k-1}^{a} \boldsymbol{A}(i) \vec{y} \otimes I_n \right) \middle| \coprod_{i=k-2}^{a} \boldsymbol{A}(i) \vec{y}} } \label{eq:law_exp}\\
    & \; = \ex{\left( \vec{y}^{\top} \prod_{i=a}^{k-2} \boldsymbol{A}(i)^{\top} \otimes I_n \right)\ex{\left( \boldsymbol{A}(k\!-\!1)^{\top} \otimes I_n \right) \var{\vect{\boldsymbol{A}(k)} }  \left( \boldsymbol{A}(k\!-\!1) \otimes I_n \right) } \left(  \coprod_{i=k-2}^{a} \boldsymbol{A}(i) \vec{y} \otimes I_n \right) } \\
    & \; = \ex{\left( \vec{y}^{\top} \prod_{i=a}^{k-2} \boldsymbol{A}(i)^{\top} \!\otimes\! I_n \right)\ex{\left( \boldsymbol{A}(k\!-\!1)^{\top} \!\otimes\! I_n \right) \left( \sum_{j=1}^n \left(\vec{e}_j \vec{e}_j^{\top} \otimes \var{\boldsymbol{A}(k) \vec{e}_j} \right)\right)  \left( \boldsymbol{A}(k\!-\!1) \!\otimes\! I_n \right) } \left(  \coprod_{i=k-2}^{a} \boldsymbol{A}(i) \vec{y} \otimes I_n \right) } \label{eq:vec_prop}\\
    & \; = \ex{\left( \vec{y}^{\top} \prod_{i=a}^{k-2} \boldsymbol{A}(i)^{\top} \otimes I_n \right)\ex{  \sum_{j=1}^n \boldsymbol{A}(k\!-\!1)^{\top}\vec{e}_j \vec{e}_j^{\top} \boldsymbol{A}(k\!-\!1) \otimes \var{\boldsymbol{A}(k) \vec{e}_j}   } \left(  \coprod_{i=k-2}^{a} \boldsymbol{A}(i) \vec{y} \otimes I_n \right) } \label{eq:kron_prop}\\
    & \; = \ex{\left( \vec{y}^{\top} \prod_{i=a}^{k-2} \boldsymbol{A}(i)^{\top} \otimes I_n \right) \left(\sum_{j=1}^n \underbrace{\ex{   \boldsymbol{A}(k\!-\!1)^{\top}\vec{e}_j \vec{e}_j^{\top} \boldsymbol{A}(k\!-\!1)}}_{S_{k-1}} \otimes \var{\boldsymbol{A}(k) \vec{e}_j} \right) \left(  \coprod_{i=k-2}^{a} \boldsymbol{A}(i) \vec{y} \otimes I_n \right) }\\
    & \; = \ex{\left( \vec{y}^{\top} \prod_{i=a}^{k-3} \boldsymbol{A}(i)^{\top} \otimes I_n \right) \left(\sum_{j=1}^n \underbrace{\ex{ \boldsymbol{A}(k\!-\!2)^{\top} S_{k-1} \boldsymbol{A}(k\!-\!2)}}_{S_{k-2}} \otimes \var{\boldsymbol{A}(k) \vec{e}_j} \right) \left(  \coprod_{i=k-2}^{a} \boldsymbol{A}(i) \vec{y} \otimes I_n \right) } \\
    & \qquad \vdots \\
    & \; = \left( \vec{y}^{\top} \otimes I_n \right) \left(\sum_{j=1}^n \ex{ \boldsymbol{A}(a)^{\top} S_{a+1} \boldsymbol{A}(a)} \otimes \var{\boldsymbol{A}(k) \vec{e}_j} \right) \left( \vec{y} \otimes I_n \right)      
\end{align}
\end{subequations}
\hrulefill
\setcounter{equation}{24}
\begin{equation} \label{eq:std_gx}
\begin{split}
    \std{\vec{G}_{k}\boldsymbol{x}(k)}^2 = & \; 
    \vec{G}_{k} \var{\coprod_{i=k}^{0} \boldsymbol{A}(i) \vec{x}(0)} \vec{G}_{k}^{\top} + \vec{G}_{k} (\vec{U}^{\top}\mathcal{B}^{\top} \otimes I_{n})\var{ \sum_{j=1}^{Nn} \left( \vec{e}_j \otimes \boldsymbol{\mathcal{C}}_k \vec{e}_j \right) } (\mathcal{B}\vec{U} \otimes I_{n})  \vec{G}_{k}^{\top} \\
    & \; \; + 2 \vec{G}_{k} (\vec{x}(0)^{\top} \otimes I_{n}) \sum_{j=1}^n \sum_{m=1}^{Nn}  \left( \vec{e}_j \vec{e}_m^{\top} \otimes \cov{ \coprod_{i=k}^{0} \boldsymbol{A}(i) \vec{e}_j }{\boldsymbol{\mathcal{C}}_k \vec{e}_m }\right) (\mathcal{B}\vec{U} \otimes I_{n})  \vec{G}_{k}^{\top}
\end{split}
\end{equation}
\hrulefill
\setcounter{equation}{19}
\end{figure*}

\begin{equation} \label{eq:cov_ck_eij}
\begin{split}
    & \cov{ \coprod_{i=k}^{a} \boldsymbol{A}(i)  \vec{e}_j}{\coprod_{i=k}^{b} \boldsymbol{A}(i)  \vec{e}_m } \\
    & \ = \ex{\prod_{i=k}^b \boldsymbol{A}(i) \left(\prod_{i=b-1}^a \ex{\boldsymbol{A}(i)} \right) \vec{e}_{i^{\ast}} \vec{e}_{j^{\ast}}^{\top} \prod_{i=b}^k \boldsymbol{A}(i)^{\top}} \\
    & \ \ - \prod_{i=k}^a \ex{\boldsymbol{A}(i)} \vec{e}_{i^{\ast}} \vec{e}_{j^{\ast}}^{\top} \prod_{i=b}^k \ex{ \boldsymbol{A}(i)^{\top}}  
\end{split}
\end{equation}
by \eqref{eq:ck_ej} and \eqref{eq:ex_quad}. This formula can easily be modified for when $a \geq b$.

Now, we expand the variance term $\var{\vec{G}_{k}\boldsymbol{x}(k)}$ as
\begin{subequations} \label{eq:var_at_k}
\begin{align}
    & \var{\vec{G}_{k}\boldsymbol{x}(k)}  \\
    & \; = \vec{G}_{k} \var{\coprod_{i=k}^{0} \boldsymbol{A}(i) \vec{x}(0)} \vec{G}_{k}^{\top} + \vec{G}_{k} \var{ \boldsymbol{\mathcal{C}}_k\mathcal{B}\vec{U}} \vec{G}_{k}^{\top}  \nonumber \\
    & \qquad + 2 \vec{G}_{k} \cov{\coprod_{i=k}^{0} \boldsymbol{A}(i) \vec{x}(0)} {\boldsymbol{\mathcal{C}}_k\mathcal{B}\vec{U}} \vec{G}_{k}^{\top} 
\end{align}
\end{subequations}
Using the formulas \eqref{eq:ex_quad}-\eqref{eq:var_ak_quad}, $\var{\coprod_{i=k}^{0} \boldsymbol{A}(i) \vec{x}(0)}$ can be found by substituting $a=0$ and $\vec{y} = \vec{x}(0)$. Next, we expand the expression for $\var{ \boldsymbol{\mathcal{C}}_k\mathcal{B}\vec{U}}$ as 
\begin{subequations}
\begin{align}
    & \var{ \boldsymbol{\mathcal{C}}_k \mathcal{B}\vec{U}} \\
    & \; = (\vec{U}^{\top}\mathcal{B}^{\top} \otimes I_{n})\var{ \vect{\boldsymbol{\mathcal{C}}_k} } (\mathcal{B}\vec{U} \otimes I_{n})  \\
    & \; = (\vec{U}^{\top}\mathcal{B}^{\top} \otimes I_{n})\var{ \sum_{j=1}^{Nn} \left( \vec{e}_j \otimes \boldsymbol{\mathcal{C}}_k \vec{e}_j \right) } (\mathcal{B}\vec{U} \otimes I_{n})  \label{eq:temp1}
\end{align}
\end{subequations}
and we can expand the variance term in \eqref{eq:temp1} as
\begin{subequations}
\begin{align}
    & \var{ \sum_{j=1}^{Nn} \left( \vec{e}_j \otimes \boldsymbol{\mathcal{C}}_k\vec{e}_j \right) } \\
    & \ =  \sum_{j=1}^{Nn} \left(\vec{e}_j\vec{e}_j^{\top} \otimes \var{ \boldsymbol{\mathcal{C}}_k\vec{e}_j} \right) \label{eq:var_cbu}\\
    & \ \ + \sum_{j=1}^{Nn}\sum_{\substack{m=1\\ j\neq m}}^{Nn} \left(\vec{e}_j \vec{e}_m^{\top}  \otimes \cov{  \boldsymbol{\mathcal{C}}_k \vec{e}_j}{\boldsymbol{\mathcal{C}}_k \vec{e}_m } \right) \label{eq:cov_cbu}
\end{align}
\end{subequations}
Hence, we can find the value of \eqref{eq:var_cbu} via \eqref{eq:ex_quad}-\eqref{eq:var_ak_quad}. Similarly, we can find the value of \eqref{eq:cov_cbu} via \eqref{eq:ck_ej} and \eqref{eq:cov_ck_eij}. Finally, we expand the covariance term,
\begin{subequations}
\begin{align}
    & \cov{\coprod_{i=k}^{0} \boldsymbol{A}(i) \vec{x}(0)} {\boldsymbol{\mathcal{C}}_k \mathcal{B}\vec{U}} \\
    & \; =  (\vec{x}(0)^{\top} \!\otimes\! I_{n}) \cov{\!\vect{\coprod_{i=k}^{0} \boldsymbol{A}(i)}\!} {\vect{\boldsymbol{\mathcal{C}}_k}\!}\! (\mathcal{B}\vec{U} \!\otimes\! I_{n}) \\
    & \; =  (\vec{x}(0)^{\top} \otimes I_{n}) \\
    & \quad \times \sum_{j=1}^n \sum_{m=1}^{Nn}  \left( \vec{e}_j \vec{e}_m^{\top} \otimes \cov{ \coprod_{i=k}^{0} \boldsymbol{A}(i) \vec{e}_j }{\boldsymbol{\mathcal{C}}_k \vec{e}_m }\right) \nonumber \\
    & \quad \times (\mathcal{B}\vec{U} \otimes I_{n}) \nonumber
\end{align}
\end{subequations}
which has a closed form via \eqref{eq:ck_ej} and \eqref{eq:cov_ck_eij}. Hence, $\std{\vec{G}_{k}\boldsymbol{x}(k)}$ has the closed form \eqref{eq:std_gx} which can be formatted as a 2-norm.  

\setcounter{equation}{25}

By inserting \eqref{eq:exp_at_k}-\eqref{eq:var_at_k} into \eqref{eq:target_reform}, we see that $\lambda_k$ and the inclusion of $\vec{U}$ in $ \std{\vec{G}_{k} \bvec{x}(k)}$ form a biconvex constraint \cite{Gorski2007}. A biconvex problem has the following form:
\begin{subequations} \label{eq:biconvex}
\begin{align}
    \min_{x,y} & \; f(x,y) \\
    \mathrm{s.t.} & \; g_i(x,y) \leq 0 \quad \forall i \in \N 
\end{align}
\end{subequations}
where $x \in X \subseteq \R^n$, $y \in Y \subseteq \R^m$, $f(\cdot, \cdot): \R^n \times \R^m \rightarrow \R$ and $g_i(\cdot, \cdot): \R^n \times \R^m \rightarrow \R$ are convex when optimizing over one parameter while holding the other constant \cite{Gorski2007}. 

For completeness, we include Algorithm \ref{algo:acs} to demonstrate one method of solving \eqref{eq:biconvex} via the well known alternate convex search method \cite{Leeuw1994}. While this method cannot guarantee global optimally, Lemma \ref{lem:2} guarantees any solution will be a feasible solution to Problem \ref{prob:1}. We note that this method can be sensitive to chosen initial conditions. However, users may opt to utilize a grid search approach to find the initial conditions that produce the most optimal solution.  

\begin{algorithm}
    \DontPrintSemicolon
    \caption{Computing solutions to \eqref{eq:biconvex} with alternate convex search}
	\label{algo:acs}
	\textbf{Input}: Feasible initial condition for $\vec{y}$, denoted $\vec{y}^\ast$, maximum number of iterations $n_{max}$.
	\;
	\textbf{Output}: Solution to \eqref{eq:biconvex}, $(\vec{x}^\ast, \vec{y}^\ast)$ \;
	\For{$i = 1 $ to $n_{max}$}{
        Solve \eqref{eq:biconvex} assuming $\vec{y} = \vec{y}^\ast$; Set $\vec{x}^\ast = \vec{x}$ \;
        Solve \eqref{eq:biconvex} assuming $\vec{x} = \vec{x}^\ast$; Set $\vec{y}^\ast = \vec{y}$ \;
        \textbf{If} Solutions converged \textbf{break}
    }
\end{algorithm}

\section{Results} \label{sec:results}

We consider the following scenario: a two-bus electric grid with two thermal generation units, one stochastic wind power plant, and a stochastic load.
All computations were done on a 1.80GHz i7 processor with 16GB of RAM, using MATLAB, CVX \cite{cvx} and Mosek \cite{mosek}. All code is available at \url{https://github.com/unm-hscl/shawnpriore-time-varying-plant}.

\subsection{Power Generation}

\begin{figure}[b]
    \centering
    \includegraphics[width=0.7\columnwidth]{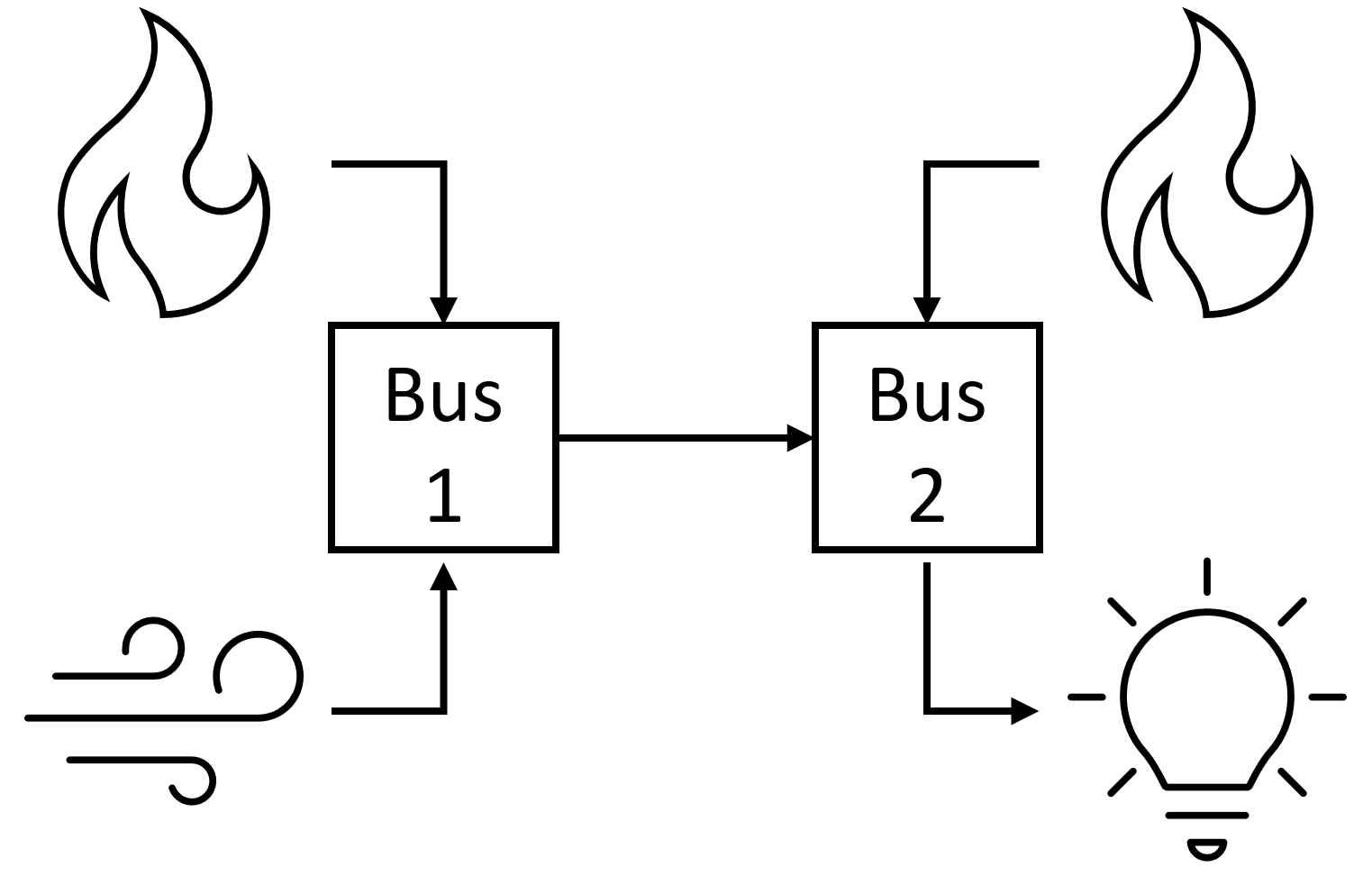}
    \caption{Two bus network with renewable source generation.}
    \label{fig:energy}
\end{figure}

Operation of power grids with high penetration of renewable energy sources are known to be a challenging task, largely because of the stochasticity inherent to intermittent sources of power generation (e.g., solar and wind generation units). Although thermal generation trends are predictable \cite{Ummels2007}, state-of-the-art modeling efforts are often unable to effectively capture the fundamentally erratic and stochastic nature of wind energy production \cite{Chen2010}.  Indeed, this is an active area of research in power grids, because of the challenges that these intermittent sources of power can create for the reliable and stable operation of power grids.

Figure \ref{fig:energy} shows a sketch of the prototypical system we consider, a small, two bus network with both wind and thermal generation, and a single load. 

We model the system with the LTV dynamics \cite{Summers2014}:
\begin{equation} \label{eq:energy}
\begin{split}
    \bvec{x}(k+1) = & \;
    \underbrace{\begin{bmatrix}
        0 & 0 & 0 & 0 & 0 & 0\\
        0 & 0 & 0 & 0 & 0 & 0\\
        0 & 0 & 1 & 0 & 0 & 0\\
        0 & 0 & \boldsymbol{\gamma} (k)^3 & 0 & 0 & 0 \\
        0 & 0 & 0 & 0 & 1 & 0 \\
        0 & 0 & 0 & 0 & \boldsymbol{\beta} (k) & 0
    \end{bmatrix}}_{\boldsymbol{A}(k)}
    \underbrace{\begin{bmatrix}
        P_1 \\ P_2 \\ C_{W} \\ P_{W} \\ C_{L} \\ L_k
    \end{bmatrix}}_{\bvec{x}(k)} \\
    & \; +
    \begin{bmatrix}
        I_2 \\ 0_{4 \times 2}
    \end{bmatrix}
    \underbrace{\begin{bmatrix}
        u_1(k) \\
        u_2(k)
    \end{bmatrix}}_{\vec{u}(k)}
\end{split}
\end{equation}
where $P_1, P_2$, and by extension $u_1(k)$ and $u_2(k)$ are the power generated (in MW) by the thermal generators connected to Bus 1 and 2, respectively, $C_{W}$ is a multiplier to convert cubed wind speed in m$^{3}\cdot$s${}^{-3}$ to MW, $P_W$ is the actual power generated from the wind farm, $C_{L}$ is the maximum load requirement in MW, and $L_k$ is the actual load. We presume $C_{L} = 1,600$ MW, and the wind farm houses 100 wind turbine generators with a blade length of 65 m. For standard air density of 1.225 kg$\cdot$m${}^{-3}$, $C_{wind} = 0.8130$ MW$\cdot$s${}^3 \cdot$m${}^{-3}$. All other initial conditions are set to 0 without loss of generality. Here, the random variables are  $\boldsymbol{\gamma} (k) \sim Weibull(5,30)$, and $\boldsymbol{\beta} (k) \sim Beta(50,50)$, and are presumed independent. Here, $\boldsymbol{\gamma} (k)$ represents the wind speed in m$\cdot$s${}^{-1}$ at time $k$ and is presumed to be consistent for all wind turbines.

The two thermal generators have a maximum nominal injection of 600 MW and must maintain at least 10\% of the maximum nominal injection to remain on, hence, 
\begin{equation}
    \mathcal{U} = \left\{ \vec{u}(k) \middle| \begin{bmatrix}
        -I_2 \\ I_2
    \end{bmatrix} \vec{u}(k) \leq \begin{bmatrix}
        60 \cdot 1_{2}\\
        600 \cdot 1_{2}
    \end{bmatrix}\right\}
\end{equation}
We presume that supply must meet demand and the power transmission line between buses has a maximum rating of 900 MW. Hence, the target set for each time step is defined by the inequality
\begin{equation}
     \underbrace{\begin{bmatrix}
        -1 & -1 & 0 & -1 & 0 & 1 \\
        1 & 0 & 0 & 1 & 0 & 0
    \end{bmatrix}}_{G_k} \bvec{x}(k) \leq \underbrace{\begin{bmatrix}
        0 \\
        900
    \end{bmatrix}}_{\vec{h}_k}
\end{equation}
Since the wind speed and load are stochastic, we consider the target constraint in a probabilistic manner and require they must hold with probability $1-\alpha$. The optimization cost is presumed to be
\begin{equation}
\begin{split}
    & J(\vec{u}(1), \dots, \vec{u}(N)) = \\
    & \; \sum_{k=1}^N \vec{u}^{\top}(k) \begin{bmatrix}
        0.05 & 0 \\ 0 & 0.10
    \end{bmatrix}\vec{u}(k) +  \begin{bmatrix}
        30 & 60
    \end{bmatrix}\vec{u}(k).
\end{split}
\end{equation}
is the cash expenditure for running the thermal generators. Here, the cost of running the thermal generator on Bus 2 is more costly than Bus 1.

From this construction, we observe that
\begin{subequations}
\begin{align}
    \prod_{i=k}^0\boldsymbol{A}(i) = & \; \boldsymbol{A}(k) \\
    \boldsymbol{A}(k)B = & \; 0_{n\times m} \quad  \forall k \in \N_{[0,N]}
\end{align}
\end{subequations}
Hence, the expectation and variance terms, \eqref{eq:exp_at_k}-\eqref{eq:var_at_k}, simplify to
\begin{subequations} \label{eq:expvar}
\begin{align}
    \ex{\vec{G}_{k}\bvec{x}(k+1)} = & \; \vec{G}_{k} \ex{\boldsymbol{A}(k)} \vec{x}(0)  + \vec{G}_{k}B(k)\vec{u}(k) \\
    \var{\vec{G}_{k}\boldsymbol{x}(k+1)} = 
    & \; \vec{G}_{k}\var{\boldsymbol{A}(k) \vec{x}(0)} \vec{G}_{k}^\top
\end{align}
\end{subequations}
The constraint \eqref{eq:target_reform} simplifies to a linear constraint for all time steps. Further, since neither term changes as a function of the time step, the optimal solution will have the same controller for all time steps. Thus, we only need to solve Reformulation \ref{prob:2} for one time step. From \eqref{eq:expvar}, we can easily find the expectation and variance of our constraints:
\begin{subequations}
\begin{align}
    \ex{\vec{G}_1  \bvec{x}(k\!+\!1)} = & \; 0.5 \cdot C_L - 118.9188 \cdot C_W \\
    & \ - \vec{u}_1(k) - \vec{u}_2(k) \nonumber\\
    \ex{\vec{G}_2 \bvec{x}(k\!+\!1)}
    = & \; 118.9188 \cdot C_W + \vec{u}_1(k)  
\end{align}
\end{subequations}
and
\begin{subequations}
\begin{align}
    \var{\vec{G}_1 \bvec{x}(k)}
    = &  \; 204.6946 \cdot C_{W}^2  + 0.0025 \cdot C_{L}^2 \\
    \var{\vec{G}_2 \bvec{x}(k)} = &  \; 204.6946 \cdot C_{W}^2 
\end{align}
\end{subequations}

 It is easy to show that the probability density function of $\boldsymbol{\gamma}(k)^3$ is log-concave via substitution. By \cite{Ibragimov1956}, $\boldsymbol{\gamma}(k)^3$ is strongly unimodal. Beta distributions with both parameters $\geq 1$ are also strongly unimodal. As strong unimodal distributions are closed under convolution, we guarantee the constraints are unimodal \cite{Bertin1997}.  

We compare the proposed methodology with scenario approach \cite{Campi2008}. As the scenario approach relies on samples of the random state matrix, it can only guarantee constraint satisfaction up to a set confidence level. For fair comparison between methods, we set the confidence level, $1-\beta$, to 0.999. We compute the number of samples, $N_S$, required to achieve this confidence level as \cite{Campi2008}
\begin{equation}
    N_S \geq \frac{2}{\alpha}\left(\log\left(\frac{1}{\beta} \right) +2 \right)
\end{equation}
corresponding to 112 samples for $1-\alpha=0.84$ and 1,781 samples for $1-\alpha=0.99$.

In Figures \ref{fig:power_cost} and \ref{fig:power_time}, we compare the optimal cost and solve time of our approach to the scenario approach. We consider discrete values of $1-\alpha \in [0.84, 0.99]$, and evaluate each approach for each value.  As shown in Figure \ref{fig:power_cost}, for lower safety probabilities, our approach has a lower cost, however, as the safety probability increases, the conservatism of our proposed approach is evident in the higher cost. (We note the proposed method was not able to find a solution at $1-\alpha =0.99$ as the admissible input set was too constraining to find a solution.) However, the solve time of our method is superior to the scenario approach for all safety probabilities for which a feasible solution was found.  In contrast, the solve time of the scenario approach appears to grow exponentially as the safety probability increases.  In considering between the two methods at high safety thresholds, the tradeoff between cost and solve time may inform choice of method. 

\begin{figure}
    \centering
    \includegraphics[width=0.8\columnwidth]{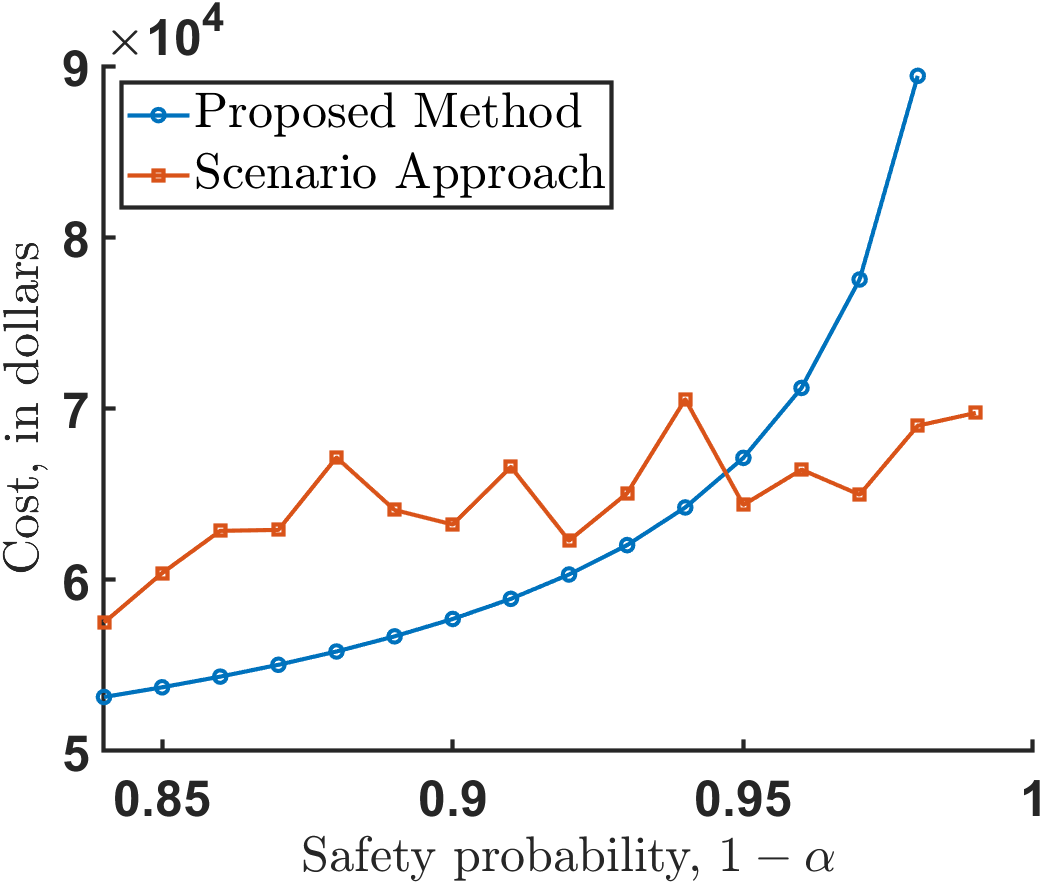}
    \caption{Comparison of optimal cost, $J$, between proposed method and scenario approach. The optimization problem was only infeasible with the proposed method at $1-\alpha =0.99$, and this is denoted by a missing value at this safety probability. The proposed method has a lower cost than the scenario approach for lower safety probabilities, and a higher cost for higher safety probabilities.}
    \label{fig:power_cost}
\end{figure}
\begin{figure}
    \centering
    \includegraphics[width=0.8\columnwidth]{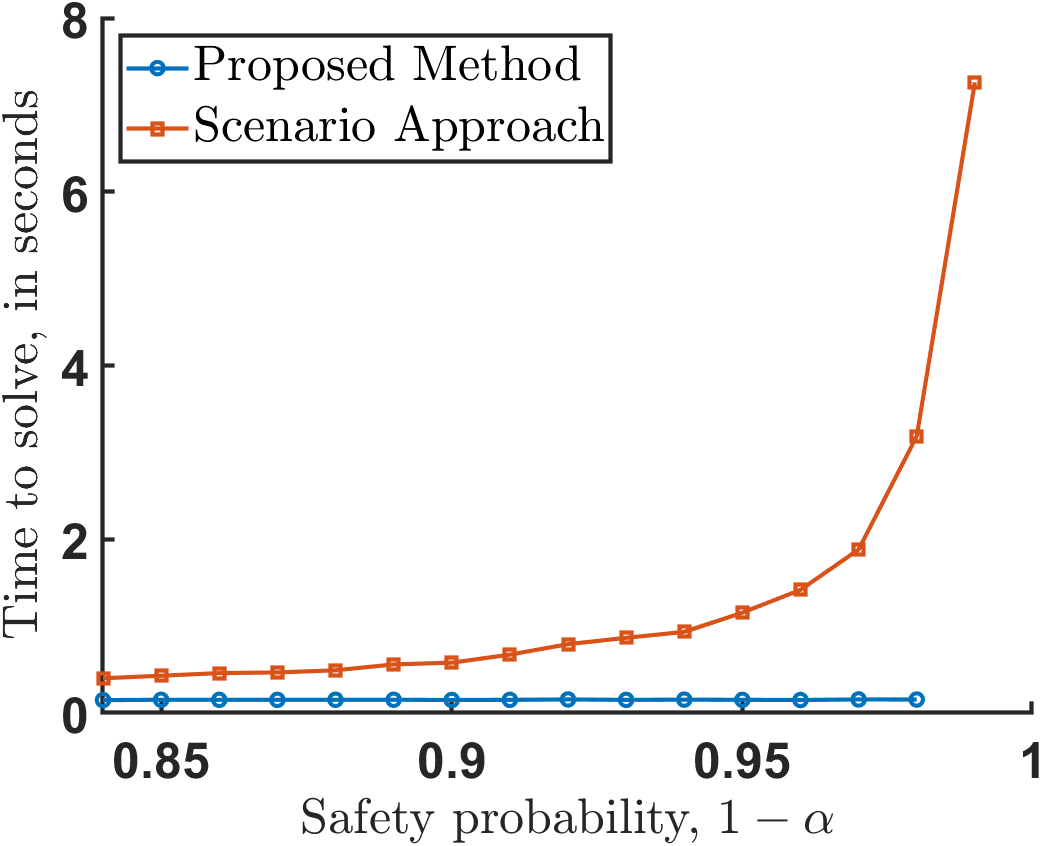}
    \caption{Comparison of time needed to find the optimal controller between proposed method and scenario approach. The optimization problem was only infeasible with the proposed method at $1-\alpha =0.99$, and this is denoted a missing value at this safety probability. The proposed method has a near constant time to solve where as the scenario approach appears to grow exponentially.}
    \label{fig:power_time}
\end{figure}

Lastly, we note that potential extension of this approach to more complex grid architectures could exploit the 
fact that the Problem 2 can be solved via a sequence of linear programs, meaning that efficient scaling would be possible.  Additionally, modeling choices in which hard constraints are cast as probabilistic constraints with high safety likelihoods may incur feasibility issues due to the conservatism inherent to the Vysochanskij–Petunin inequality.

\section{Conclusion} \label{sec:conclusion}

We proposed a framework for solving stochastic optimal control problems for systems with random plant parameters subject to polytopic target set chance constraints. Our approach relies on the one-sided Vysochanskij–Petunin inequality to reformulate the joint chance constraints into a series of individual chance constraints. We have shown that these new constraints typically result in a biconvex optimization problem and outlined the alternate convex search approach to solve them. We demonstrated our method for an stochastic multi-input power generation model and compared our results with the scenario approach. We showed that our method performed two orders of magnitude faster and for some safety thresholds resulted in a lower optimal cost. 

\bibliographystyle{ieeetr}
\bibliography{main}
\end{document}